\documentclass[11pt]{article}
\usepackage{amsmath,amssymb,amsfonts,amscd,xspace,pifont}
\usepackage{epsfig,amsfonts,color,soul,url}
\usepackage{natbib}
\usepackage{multirow}
\usepackage{etoolbox}
\usepackage{float}
\usepackage{booktabs}
\usepackage{graphicx}
\usepackage{caption}
\usepackage{subcaption}

\bibpunct{(}{)}{;}{a}{}{,}

\newtheorem{theorem}{Theorem}

\newtheorem{lemma}{Lemma}

\newenvironment{proof}[1][Proof]{\begin{trivlist}
  \item[\hskip \labelsep {\bfseries #1}]}{\end{trivlist}}
\definecolor{red}{rgb}{1,0,0}
\definecolor{blue}{rgb}{0,0,1}
\definecolor{green}{rgb}{0,0.6,0.4}

\def\blue{\color{blue}}

\newcommand{\V}[1]{\ensuremath{\boldsymbol{#1}}\xspace}

\newcommand{\diag}{\F{diag}}

\def\E{\mathop{\rm E}\nolimits}
\def\P{\mathbb{P}}
\def\Var{\mathop{\rm Var}\nolimits}

\def\diag{\mathop{\rm diag}\nolimits}

\def\argmax{\mathop{\rm argmax}\nolimits}
\newcommand{\bea}{\begin{eqnarray}}
\newcommand{\eea}{\end{eqnarray}}
\newcommand{\beaa}{\begin{eqnarray*}}
\newcommand{\eeaa}{\end{eqnarray*}}
\newcommand{\bi}{\begin{itemize}}
\newcommand{\ei}{\end{itemize}}

\newcommand*{\QEDB}{\hfill\ensuremath{\square}}%

\begin{document}

\title{Simultaneous Detection of Signal Regions Using Quadratic Scan Statistics With Applications in Whole Genome Association Studies}
\author{Zilin Li, Yaowu Liu and Xihong Lin  \thanks{Zilin Li is a postdoctoral fellow in the Department of Biostatistics at Harvard T.H. Chan School of Public Health({\em li@hsph.harvard.edu}). Yaowu Liu is a postdoctoral fellow in the Department of Biostatistics at Harvard T.H. Chan School of Public Health ({\em yaowuliu@hsph.harvard.edu}). Xihong Lin is Professor of Biostatistics at Harvard T.H. Chan School of Public Health and Professor of Statistics at Harvard University ({\em xlin@hsph.harvard.edu}). This work was supported by grants R35-CA197449, U19CA203654  and P01-CA134294 from the National Cancer Institute, U01-HG009088 from the National Human Genome Research Institute, and R01-HL113338 from the National Heart, Lung, and Blood Institute.}}

\date{}

\maketitle
\thispagestyle{empty}
\baselineskip=20pt

\begin{abstract}
\noindent{We consider in this paper detection of signal regions associated with disease outcomes in whole genome association studies. Gene- or region-based methods have become increasingly popular in whole genome association analysis as a complementary approach to traditional individual variant analysis. However, these methods test for the association between an outcome and the genetic variants in a pre-specified region, e.g., a gene. In view of massive intergenic regions in whole genome sequencing (WGS) studies, we propose a computationally efficient quadratic scan (Q-SCAN) statistic based method to detect the existence and the locations of signal regions by  scanning the genome continuously. The proposed method accounts for the correlation (linkage disequilibrium) among genetic variants, and allows for signal regions to have both causal  and neutral variants, and the effects of signal variants to be in different directions. We study the asymptotic properties of the proposed Q-SCAN statistics. We derive an empirical threshold that controls for the family-wise error rate, and show that under regularity conditions the proposed method consistently selects the true signal regions. We perform simulation studies to evaluate the finite sample performance of the proposed method. Our simulation results show that the proposed procedure outperforms the existing methods, especially when signal regions have causal variants whose  effects are in different directions, or are contaminated with neutral variants. We illustrate Q-SCAN by analyzing the WGS data from the Atherosclerosis Risk in Communities (ARIC) study.
}

\vspace{.2in}
\end{abstract}

\newpage
\pagestyle{plain}
\setcounter{page}{1}

\section{Introduction}\label{sec:introduction}
An important goal of human genetic research is to identify the genetic basis for human diseases or traits. Genome-Wide Association Studies (GWAS) have been widely used to dissect the genetic architecture of complex diseases and quantitative traits in the past ten years. GWAS uses an array technology that genotypes millions of Single Nuclear Polymorphisms (SNPs) across the genome, and aims at identifying SNPs that are associated with traits or disease outcomes. GWAS has been successful for identifying thousands of common genetic variants putatively harboring susceptibility alleles for complex diseases \citep{visscher2012five}.  However, these common variants only explain a small fraction of heritability \citep{manolio2009finding} and the vast majority of variants in the human genome are rare \citep{10002010map}. A rapidly increasing number of Whole Genome Sequencing (WGS) association studies are being conducted to identify susceptible rare variants, for example the Genome Sequencing Program (GSP) of the National Human Genome Research Institute, and the Trans-Omics for Precision Medicine (TOPMed) Program of the National Heart, Lung, and Blood Institute.

A limitation of GWAS is that it only genotypes common variants.  A vast majority of variants in the human genome are rare \citep{10002010map, tennessen2012evolution}. Whole genome sequencing studies allow studying rare variant effects. Individual variant analysis that commonly used in GWAS is however not applicable for analyzing rare variants in WGS due to a lack of power
\citep{bansal2010statistical, kiezun2012exome, lee2014rare}. Gene-based tests, as an alternative to the traditional single variant test, have become increasingly popular in recent years in GWAS analysis \citep{li2008methods, madsen2009groupwise, wu2010powerful}. Instead of testing each SNP individually, these gene based tests evaluate the cumulative effects of multiple variants in a gene, and can boost power when multiple variants in the gene are associated with a disease or a trait \citep{han2009rapid,wu2010powerful}. There is an active recent literature on rare variant analysis methods which jointly test the effects of multiple variants in a variant set, e.g., a genomic region,  such as burden tests \citep{morgenthaler2007strategy, li2008methods, madsen2009groupwise}, and  non-burden tests \citep{wu2011rare, neale2011testing,lin2011general, lee2012optimal}, e.g., Sequence Kernel Association Test (SKAT) \citep{wu2011rare}. The primary limitation of these gene-based tests is that it needs to pre-specify genetic regions, e.g., genes, to be used for analysis. Hence these existing gene-based approaches are not directly applicable to WGS data, as analysis units are not well defined across the genome, because of a large number of intergenic regions. It is of substantially interest to scan the genome continuously to identify the sizes and locations of signal regions.

Scan statistics \citep{naus1982approximations} provide an attractive framework to scan the whole genome continuously for detection of signal regions in whole genome sequencing data. The classical fixed window scan statistics allow for overlapping windows using a moving window of a fixed size, which ``shifts forward" a window with a number of variants at a time and searches for  the  windows containing signals. A limitation of this approach is that the window size needs to be pre-specified in an ad hoc way. In cases where multiple variants are independent in a sequence, \cite{sun2006scan} proposed a region detection procedure using a scan statistic that aggregates the $p$-values of individual variant tests. However this is not applicable to WGS due to the linkage disequilibrium (LD), i.e,. correlation, among nearby variants. Recently, \cite{mccallum2015empirical} proposed likelihood-ratio-based scan statistic procedure to refine disease clustering region in a gene, but not for testing associations across the genome. Furthermore, this method does not allow for covariates adjustment (e.g., age, sex and population structures), and can be only used for binary traits.

The mean-based scan statistic procedures have been used  in DNA copy number analysis. Assuming  all variants are signals with the same mean in signal regions, several authors have proposed to use the mean of marginal test statistics in each candidate region as a scan statistic. Specifically, \cite{arias2005near} proposed a likelihood ratio-based  mean scan procedure in the presence of only one signal region. \cite{zhang2010detecting} described an analytic approximation to the significance level of this scan procedure, while \cite{jeng2010optimal} showed this procedure is asymptotically optimal in the sense that it separates the signal segments from the non-signals if it is possible to consistently detect the signal segments  by any identification procedure. This setting is closely related to the change-point detection problem. \cite{olshen2004circular} developed an iterative circular binary segmentation procedure to detect change-points, whereas \cite{zhang2007modified,zhang2012model} proposed a BIC-based model selection criterion for estimating the number of change-points. However,  the key assumption of these mean scan procedures that all observations have the same signals  in signal regions generally does not hold in genetic association studies.

Indeed, although the mean based scan statistics are useful for copy number analysis, these procedures have several limitations for detecting signal regions in whole genome array and sequencing association studies. Specifically, they will lose power due to signal cancellation in the presence of both trait-decreasing and trait-increasing genetic variants, or the presence of both causal and neural variants in a signal region. Both situations are common in practice.  Besides, these procedures assume the individual variant test statistics are independent across the whole genome. However, in practice, the variants in a genetic region are correlated due to linkage disequilibrium (LD).

In this paper, we propose a quadratic scan statistic based procedure (Q-SCAN) to detect the existence and locations of signal regions in whole genome association studies by scanning the whole genome  continuously. Our procedure can consistently detect an entire signal segment in the presence of both trait-increasing and trait-decreasing variants and mixed signal and neutral variants. It also accounts for the correlation (LD) among the variants when scanning the genome. We derive an empirical threshold that controls the family-wise error rate. We study the asymptotic property of the proposed scan statistics, and show that the proposed procedure can consistently select the exact true signal segments under some regularity conditions. We propose a computationally efficient searching algorithm for the detection of multiple non-overlapping signal regions.

We conduct simulation studies to evaluate the finite sample performance of the proposed procedure, and compare it with several existing methods. Our results show that, the proposed scan procedure outperforms the existing methods in the presence of weak causal and neutral variants, and both trait-increasing and trait-decreasing variants in signal regions. The advantage of the proposed method is more pronounced in the presence of the correlation (LD) among the variants in  signal regions. We  applied  the proposed procedure to the analysis of  WGS lipids data from the Atherosclerosis Risk in Communities (ARIC) study to identify genetic regions which are associated with lipid traits.

The remainder of the paper is organized as follows. In Section 2, we introduce the hypothesis testing problem and describe our proposed scan procedure and a corresponding algorithm to detect multiple signal regions. In Section 3, we present the asymptotic properties of the scan statistic, as well as the statistical properties of  identifiable regions. In Section 4, we compare the performance of our procedure with other scan statistic procedures in simulation studies. In Section 5, we apply the proposed scan procedure to analyze WGS data from the ARIC study. Finally, we conclude the paper with discussions in Section 6. The proofs are relegated to the Appendix.

\section{The Statistical Model and the Quadratic Scan Statistics for Signal Detection}\label{sec:model}
\subsection{Summary Statistics of Individual Variant Analysis Using Generalized Linear Models}
Suppose that the data are from  $n$ subjects. For the $i$th subject ($i=1,\cdots, n$), $Y_i$ is an outcome, $\V X_i=(X_{i1},\dots,X_{iq})^T$ is a vector of $q$  covariates, and   $G_{ij}$  is the $j$th of $p$ variants in the genome. One constructs individual variant test statistics in GWAS and WGS studies  by regressing $Y_i$ on each variant $G_{ij}$ adjusting for the covariates $\V X_i$. Conditional on $(\V X_i,\V G_{ij})$, $Y_i$ is assumed to follow a distribution in the exponential family  with the density $f(Y_i) = \exp\{Y_i\theta_i-b(\theta_i)/a_i(\phi) + c(Y_i,\phi) \}$, where $a(\cdot)$, $b(\cdot)$ and $c(\cdot)$ are some known functions, and $\theta_i$ and $\phi$ are the canonical parameter and the dispersion parameter, respectively \citep{mccullagh1989generalized}. Denote by $\eta_{i}=\E(Y_i|\V X_i, G_{ij})=b'(\theta_i)$.  The test statistic for the $j$th variant is constructed using the following Generalized Linear Model (GLM) \citep{mccullagh1989generalized}
\begin{equation*}
g(\eta_{i})=\V X_i^T\V\alpha+G_{ij}\beta_j,
\end{equation*}
where $g(\cdot)$ is a monotone link function. For simplicity, we assume $g(\cdot)$ is a canonical link function. The variance of $Y_i$ is $var(Y_i)=a_i(\phi)v(\eta_{i})$, where $v(\eta_i)=b^{\prime\prime}(\theta_i)$ is a variance function.

Let $\hat{\eta}_{0i}=g^{-1}(\V X_i^T\hat{\V\alpha})$, where $\hat{\V\alpha}$ is the Maximum Likelihood Estimator (MLE) of $\V\alpha$, and $\hat{\phi}$ is the MLE of $\phi$, both under the global null model of $g(\eta_i)=\V X_i^T\V\alpha$. Assume $\V\Lambda=\diag\big \{a_1(\hat \phi)v(\hat{\eta}_{01}),\dots,a_n(\hat\phi)v(\hat{\eta}_{0n})\big \}$ and $\V P = \V \Lambda^{-1} - \V\Lambda^{-1} \V X (\V X^T \V \Lambda^{-1} \V X)^{-1}\V X^T \V \Lambda^{-1}$. The marginal score test statistic for $\beta_j$ of the $j$th variant is
\begin{equation*}
U_j = \V G_j^T (\V Y -\hat{\V\eta}_0)\big/\sqrt{n},
\end{equation*}
where $\V G_j=(G_{1j},\cdots, G_{nj})^T$ denotes the $j$th variant data of $n$ subjects, $\hat{\V \eta}_0=(\hat\eta_{10},\cdots,\hat{\eta}_{n0})^T$ and $\V Y=(Y_1,\cdots, Y_n)^T$. These individual variant test statistics are asymptotically jointly distributed as $\V U \sim N(\V\mu,\V \Sigma)$, where $\V U=(U_1,\cdots, U_p)^T$, $\V \mu=E(\V U)$. Note that $\V \mu=0$ under the global null of all $\beta_j$ being 0, and   $\V \Sigma_{jj'}$ can be estimated by
\begin{equation}\label{Sigma-hat}
\hat{\V\Sigma}_{jj'} = \V G_j^T \V P \V G_{j'}\big/n.
\end{equation}
These individual SNP summary statistics $U_j$ are often available in public domains or provided by investigators to facilitate meta-analysis of multi-cohorts.

Genetic region-level analysis has become increasingly important in GWAS and WGS  rare variant association studies \citep{li2008methods, lee2014rare}.  The existing region-based tests require pre-specification of regions using biological constructs, such as genes. For a given region,  region-level analysis aggregates the marginal individual SNP test statistics $U_j$ across the variants in the region to test for the significance of the region
\citep{li2008methods, madsen2009groupwise,wu2011rare}. However, whole genome array and sequencing studies consist of many intergenic regions. Hence, analysis based on genes or prespecified regions of a fixed length, e.g., a moving window of 4000 basepairs, are not desirable for scanning the genome to detect signal segments. This is because region-based tests will lose power if a pre-specified region is too big or too small.  Indeed, it is of primary interest in whole genome association analysis to scan the whole genome to detect the sizes and locations of the regions that are associated with diseases and traits. We tackle this problem using the quadratic scan statistic.

\subsection{Detection of Signal Regions Using Scan Statistics}\label{subsec:signal region detection}
Let a sequence of $p$ marginal test statistics be $\V U=\{U_1,\dots,U_p\}$, where $U_i$ is the marginal test statistic at location $i$ and $p$ is the total number of locations, e.g., the total number of variants in GWAS or WGS. We assume that the sequence $\V U$ follows a multivariate normal distribution
\begin{equation}\label{statistical model}
\V U \sim N(\V\mu,\V\Sigma),
\end{equation}
where $ \V \mu$ is an unknown mean of $\V U$ and $\V \Sigma=cov(\V U)$. Under the global null hypothesis of no signal variant across the genome, we have $\V\mu=0$. Under the alternative hypothesis of non-overlapping signal regions, there exist signals at certain non-overlapping regions $I_1,\dots,I_r$ satisfying $\V\mu_{I_j} \neq 0$, where $\V\mu_{I_j}=\{\mu_i\}_{i \in I_j}$ and $j=1,\dots,r$. Note the lengths of the signal regions are allowed to be different. Besides, the signal region $I_j$ satisfies that in a large area that contains $I_j$, there is no signal point ($\mu=0$) outside $I_j$ and the edges of $I_j$ are signal points. Denote a collection of the non-overlapping signal regions by $\mathcal{I}=\{I_1,\dots,I_r\}$. Our goal is to detect whether signal segments exist, and if they do exist, to identify the location of these segments. Precisely, we wish to first test
\begin{equation}\label{testing}
H_0:\mathcal{I}=\emptyset \quad \mathrm{against} \quad H_1:\mathcal{I} \neq \emptyset,
\end{equation}
and if $H_0$ is rejected, detect each signal region in $\mathcal{I}$.

A scan statistic procedure solves the hypothesis testing problem (\ref{testing}) by using the extreme value of the scan statistics of all possible regions,
\begin{equation}\label{eq:scan statistics}
Q_{\max}=\displaystyle{\max_{L_{\min} \leq|I| \leq L_{\max}}} Q(I),
\end{equation}
where $Q(I)$ is the scan statistic for region $I$, $|I|$ denotes the number of variants in region $I$,  and $L_{\min}$ and $L_{\max}$ are the minimum and maximum variants number in searching windows, respectively. A large value of $Q_{\max}$ indicates evidence against the null hypothesis. If the null hypothesis is rejected and results in only one region, the selected signal region is $\hat{I}=\displaystyle{\argmax_{L_{\min} \leq |I|\leq L_{\max}}} Q(I)$.

\cite{jeng2010optimal} and \cite{zhang2010detecting} proposed a scan procedure based on the mean of the marginal test statistics of a candidate region (M-SCAN). The mean scan statistic for region $I$ is defined as
\begin{equation}
\label{eq:mean scan statistics}
M(I)=\sum_{i \in I} Z_i\big/\sqrt{|I|},
\end{equation}
where $Z_i = U_i / var(U_i)$ is the standardized score statistics. When the test statistics $Z_i$ are independent ($\V \Sigma=\V I_n$) with a common mean in a signal region ($\mu_{i}=\mu$ for all $i \in I_j$), \cite{arias2005near} and \cite{jeng2010optimal} showed that the mean scan procedure is asymptotically optimal in the sense that it separates the signal segments from the non-signals as long as the signal segments are detectable. However, in whole genome association studies, the assumptions that marginal tests $Z_i$ are independent and have the same mean in signal regions often do not hold. This is because, first,  marginal test statistics in  a region are commonly correlated due to the LD of variants; second, signal variants in a signal region are likely to have effects in different directions and be mixed with neutral variants in the signal regions. Hence application of the existing mean scan statistics (\ref{eq:mean scan statistics}) for detecting signal regions in whole genome association studies is likely to not only yield invalid inference due to failing to account for the correlation between the $Z_i$'s across the genome, but also more importantly lose power due to cancellation of signals in different directions in signal regions.

\subsection{The Quadratic Scan Procedure}\label{subsec:sum of chi-squared}
To overcome the limitations of the mean scan procedure, we propose a quadratic scan procedure (Q-SCAN) that selects signal regions based on the sum of quadratic marginal test statistics, which is defined as,
\begin{equation}\label{eq:quadratic scan statistics}
Q(I) = \frac{\sum_{i \in I} U_i^2  - \E(\sum_{i \in I} U_i^2)}{var(\sum_{i \in I} U_i^2)}
= \frac{\sum_{i \in I} U_i^2 - \sum_{i=1}^{|I|} \lambda_{I,i}}{\sqrt{2\sum_{i=1}^{|I|}\lambda_{I,i}^2}},
\end{equation}
where $\V\lambda_I = (\lambda_{I,1},\lambda_{I,2},\cdots,
\lambda_{I,|I|})^T$ is the eigenvalues of $\V\Sigma_I = \V G_I^T \V P \V G_I \big/ n$ and $\V\Sigma_I$ is the covariance matrix of test statistics $\V U_I$. In the presence of correlation among the test statistics $Z$'s,  the null distribution of $Q(I)$ is  a  centered mixture of chi-squares $\sum_{j=1}^{|I|} \big(\lambda_{I,j}/\sqrt{2\sum_{i=1}^{|I|}\lambda_{I,i}^2}\big) \times (\chi_{1j}^2-1)$, where the $\chi_{1j}^2$ are independent chi-square random variables with one degree of freedom, and asymptotically follows the standard normal distribution when $|I|\rightarrow \infty$. When signals have different directions, the proposed quadratic scan statistic avoids signal cancellation that will result from  using the mean scan  statistic (\ref{eq:mean scan statistics}). By using the standardization of the sum of quadratic marginal test statistics in region $I$, the scan statistics $Q(I)$ handles the different LD structure across the genome and are comparable for different region lengths.

We reject the null hypothesis (\ref{testing}) if the scan statistic of a region is larger than a given threshold $h(p,L_{\min},L_{\max},\alpha)$. The threshold $h(p,L_{\min},L_{\max},\alpha)$ is used to control the family-wise error rate at exact $\alpha$ level. If this results in only one region, the estimated signal region is $\hat{I} = \argmax_{L_{\min} \leq |I| \leq L_{\max}} Q(I)$. If this results in multiple overlapping regions, we estimated the signal region as the interval whose test statistic is greater than the threshold and achieves the local maximum in the sense that the test statistic of that region is greater than the regions that overlap with it. We propose a searching algorithm to consistently detect true signal regions in the next section.

We propose to use Monte Carlo simulations to determine $h(p,L_{\min},L_{\max},\alpha)$ empirically. Specifically, we generate samples from $N(0,\hat{\V \Sigma})$ and calculate $Q_{\max}$. We repeat this for $N$ times and use the $1-\alpha$ quantile of the empirical distribution as the data-driven threshold. Section \ref{subsec:threshold} presents details on calculating the empirical threshold.

\subsection{Searching Algorithm for Multiple Signal Regions}\label{subsec:Searching algorithm}
In general, there might be several signal regions in a whole genome. We now describe an algorithm for detecting multiple signal regions. Motivated by GWAS and WGS, we assume the signal regions are short relatively to the size of the whole genome, and are reasonably well separated. Hence intuitively, the test statistic for proper signal region estimation should achieve a local maximum. Following \cite{jeng2010optimal} and \cite{zhang2010detecting}, our proposed searching algorithm first finds all the candidate regions with the quadratic scan statistic greater than a pre-specified threshold $h(p,L_{\min},L_{\max},\alpha)$. Then we select the intervals from the candidate sets that has the largest test statistic than the other overlapped intervals in the candidate set as the estimated signal regions.  The detailed algorithm is given as follows:

\begin{description}
\item[Step 1.] Set the minimum variants number in searching window $L_{\min}$ and maximum  variants number in searching window $L_{\max}$ and calculate $Q(I)$ for the intervals with variants number between $L_{\min}$ and $L_{\max}$.
\item[Step 2.] Pick the candidate set
    $$\mathcal{I}^{(1)} = \{I :\: Q(I) > h(p,L_{\min},L_{\max},\alpha),\, L_{\min} \leq |I| \leq L_{\max} \}$$
    for some threshold $h(p,L_{\min},L_{\max},\alpha)$. If $\mathcal{I}^{(1)}\neq \emptyset$, we reject the null hypothesis, set $j=1$ and proceed with the following steps.
\item[Step 3.] Let $\hat{I}_j = \argmax_{I \in \mathcal{I}^{(j)}} Q(I)$, and update $\mathcal{I}^{(j+1)} = \mathcal{I}^{(j)} \backslash \{ I\in \mathcal{I}^{(j)}: \: I\bigcap\hat{I}_j \neq \emptyset \}$.
\item[Step 4.] Repeat Step 3 and Step 4 with $j=j+1$ until $\mathcal{I}^{(j)}$ is an empty set.
\item[Step 5.] Define $\hat{I}_1,\hat{I}_2,\cdots$ as the estimated signal regions.
\end{description}

After the test statistic $Q(I)$ is calculated for each region $L_{\min} \leq |I| \leq L_{\max}$, we can estimate the null distribution of $Q_{\max}$.  A threshold $h(p,L_{\min},L_{\max},\alpha)$ is set based on the null distribution of $Q_{\max}$. Specifically, the threshold $h(p,L_{\min},L_{\max},\alpha)$ is calculated to control for the family-wise error rate at a desirable level $\alpha$ by adjusting for multiple testing of all searched regions. Section \ref{subsec:threshold} provides detailed discussions on calculating $h(p,L_{\min},L_{\max},\alpha)$ and Section \ref{subsec:familywise error rate} discussed the bound of $h(p,L_{\min},L_{\max},\alpha)$.

Steps 3-4 are used to search for all the local maximums of the scan statistic $Q(I)$  by iteratively selecting the intervals from the candidate set with the largest scan statistics $Q(I)$, and then deleting a selected signal interval and any other intervals overlapping with it from the candidate set before moving on to select the next signal interval. Step 5 collects all the local maximums as the set of selected signal regions. The intuition of this algorithm is as follows. Since we assume the signal segments are well separated in the sequence, for a signal region, no region with variants number between $L_{\min}$ and $L_{\max}$ overlaps with more than one signal region. Thus the test statistic of a signal region $I$ is larger than the other intervals that overlap with it. It follows that a local maximum provides  good estimation of a signal region.

Selection of the minimum and maximum variants numbers in searching windows, i.e., $L_{\min}$ and $L_{\max}$, is an important issue in scan procedures. Specifically, to ensure that each signal region will be searched, $L_{\min}$ and $L_{\max}$ should be smaller and larger than the variants number in all signal regions, respectively. In the meantime, $L_{\max}$ should be smaller than the shortest gap between signal regions to ensure that no candidate region $I$ with $L_{\min} \leq |I|\leq L_{\max}$ overlaps with two or more signal regions. These two parameters $L_{\min}$ and $L_{\max}$ also determine computation complexity. A smaller range between $L_{\min}$ and $L_{\max}$ requires less computation. Instead of setting the range of moving window sizes on the basis of the number of base pairs, we specify a range of moving window sizes on the basis of the number of variants. In practice, we recommend to choose $L_{\min} = 40$ and $L_{\max} = 200$. Different from the fact that the observed rare variants number in a given window increases with a fixed number of base pairs as sample size increases, such a range specification on the basis of the number of variants is  independent of sample sizes.

\subsection{Threshold for Controlling the Family-Wise Error Rate}\label{subsec:threshold}
Although the theorems in Section \ref{sec:theory} shows that the family-wise error rate can be asymptotically controlled, it is difficult to use a theoretical threshold for an exact $\alpha-$level test in practice. The standard Bonferroni correction for multiple-testing adjustment is also too conservative for the Q-SCAN procedure, because the candidate search regions overlap with each other and the scan statistics for these regions are highly correlated. Therefore, we propose to use Monte Carlo simulations to determine an empirical threshold to control for the family-wise error rate at the $\alpha-$level. For each step, we estimate $\hat{\V\Sigma}$ using (\ref{Sigma-hat}) and generated samples from $ N(0,\hat{\V \Sigma})$. Specifically,  we first generate $\tilde{\V u} \sim N(0,\V I_n)$ and calculate the pseudo-score vector by $\tilde{\V U} = \V G \V P^{1/2} \V u/\sqrt{n}$ where $\V G = (\V G_1^T,\cdots, \V G_n^T)$ is the $p \times n$ genotype matrix, $n$ is the number of subjects in the study, and $\V P$ is the $n \times n$ projection matrix of the null GLM. Then we calculate the extreme value $Q_{\max}$ of the test statistic $Q(I)$  using (\ref{eq:scan statistics}) across the genome based on the pseudo-score $\tilde{\V U}$ and the estimated covariance matrix $\hat{\V \Sigma} = \V G \V P \V G^T/n$. We repeat this for a large number of times, e.g, $2000$ times, and use the $1 - \alpha$ quantile of the empirical distribution of $Q_{\max}$ as the empirical threshold $h(p,L_{\min},L_{\max},\alpha)$ for controlling the family-wise error rate at $\alpha$.

\section{Asymptotic Properties of the Quadratic Scan Procedure}\label{sec:theory}
In this section, we present two theoretical properties of the quadratic scan procedure. The first property shows the convergence rate of the extreme value $Q_{\max}$ and gives a bound of the empirical threshold $h(p,L_{\min},L_{\max},\alpha)$. The second property shows that, under certain regularity conditions, the quadratic scan procedure consistently detects the exact signal regions.

\subsection{Bound of the Empirical Threshold}\label{subsec:familywise error rate}
We first provide a brief summary of notation used in the paper. For any vector $\V a $, set $||\V a||_1 = \sum_i |a_i|$, $||\V a||_2 = \sqrt{\sum_i a_i^2}$ and $||\V a||_{\infty} = \sup_i |a_i|$. For two sequences of real numbers $a_p$ and $b_p$, we say $a_p \ll b_p$ or $a_p = o(b_p)$, when $\limsup a_p/b_p \rightarrow 0$. Recall that $U_i = \V G_i^T(\V Y - \hat{\V\mu})/\sqrt{n}$ is the score statistic for variant $i$ ($i=1,2,\cdots,p$.) Under the null hypothesis, $U_i \sim N(0,\sigma_i^2)$ with $\sigma_i^2 = \V G_i^T \V P \V G_i/n$ for all $i=1,2,\cdots,p$, where $\V P$ is the projection matrix in the null model. Assume there exists a constant $c>0$ such that $\sigma_i^2 \geq c$. The following theorem gives the convergence rate of $Q_{\max}$.
\begin{theorem}\label{th-multi-scan}
If the following conditions (A)-(C) hold,
\begin{itemize}
\item[(A)] $\max_{|I|=L_{\max}} ||\V\lambda_I||_{\infty} \leq K_0$, where $K_0$ is a constant,
\item[(B)] $\frac{L_{\min}}{\log(p)} \rightarrow \infty$ and $\frac{\log(L_{\max})}{\log(p)} \rightarrow 0$,
\item[(C)] $\{U_i\}_{i=1}^p$ is $M_p$-dependent and $\frac{\log(M_p)}{\log(p)} \rightarrow 0$,
\end{itemize}
then
\begin{equation*}
\frac{\max_{L_{\min} \leq |I| \leq L_{\max}} Q(I)}{\sqrt{2\log(p)}} \stackrel{p}{\rightarrow} 1.
\end{equation*}
\end{theorem}

Condition (A) holds when the lower bound of the minor allele frequency of variants is a constant that greater than 0.
In GWAS and WGS, the number of variants $p$ is large, e.g.,  from hundreds of thousands to hundreds of millions. However, $\log (p)$ grows much slower and is comparable to the length of genes and of LD blocks, so that the condition (B) of $L_{\min}$ and $L_{\max}$ is reasonable in practice. Further, since two marginal test statistics are  independent when two variants are sufficiently far apart in the genome, the assumption of $M_p$ dependence in Condition (C) is reasonable in reality.

Note the empirical threshold $h(p,L_{\min},L_{\max},\alpha)$ is the $(1-\alpha)$th quantile of $Q_{\max}$, that is,
\begin{equation*}
\P\big(Q_{\max} > h(p,L_{\min},L_{\max},\alpha)\big) = \alpha.
\end{equation*}
By Theorem \ref{th-multi-scan},  for any $\epsilon>0$, when $p$ is sufficiently large, we have $ (1-\epsilon) \sqrt{2\log(p)} \leq h(p,L_{\min},L_{\max},\alpha) \leq (1+\epsilon) \sqrt{2\log(p)}$. Next we give a more accurate upper bound of $h(p,L_{\min},L_{\max},\alpha)$.

\begin{theorem}\label{prop:thres-upper}
If conditions (A) and (B) in Theorem \ref{th-multi-scan} hold, for $p$ sufficiently large, we have
\begin{equation*}
h(p,L_{\min},L_{\max},\alpha) \leq \sqrt{2\big[\log\{p(L_{\max}-L_{\min})\} - \log(\alpha)\big]} + \frac{\sqrt{2}\big[\log\{p(L_{\max}-L_{\min})\} - \log(\alpha)\big]}{\{L_{\min}\log(p)\}^{\frac{1}{4}}}.
\end{equation*}
\end{theorem}
By Theorems \ref{th-multi-scan} and \ref{prop:thres-upper}, for $p$ sufficiently large, we give the bound of the empirical threshold as follows,
\begin{equation*}
(1-\epsilon)\sqrt{2\log(p)} \leq h(p,L_{\min},L_{\max},\alpha) \leq \sqrt{2\gamma_p} + \frac{\sqrt{2}\gamma_p}{\{L_{\min}\log(p)\}^{\frac{1}{4}}},
\end{equation*}
 where $\epsilon$ is a small constant and $\gamma_p = \log\{p(L_{\max}-L_{\min})\} - \log(\alpha) $.

\subsection{Consistency of Signal Region Detection}\label{subsec:Power analysis}
In this section, we show the results of power analysis. We first show that the proposed Q-SCAN procedure could consistently select a signal region that overlaps with the true signal region. Let $\V\mu_{I}=\{\mu_i\}_{i \in I}$ for any region $I$. Assume $I^*$ is the signal region with $\V\mu_{I^*} \neq 0$ and $L_{\min} \leq |I^*| \leq L_{\max}$. Denote the signal region by $I^* = (\tau_1^*,\tau_2^*]$ that satisfies certain regularity conditions on its norm and edges, e.g., the $L_2$ norm measuring the overall signal strength of $I^*$ is sufficiently large and the edges of $I^*$ are signal points, that is, $\mu_{\tau_1^*+1} \neq 0$ and $\mu_{\tau_2^*} \neq 0$. We also assume there is no signal point ($\mu \neq 0$) outside $I^*$ in a large area that contains $I^*$, that is, there exists $\tau \geq L_{\max}$, such that $\V\mu_{I_1} = \V\mu_{I_2} = 0$, where $I_1=(\tau_1^*-\tau,\tau_1^*]$ and $I_2 = (\tau_2^*,\tau_2^*+\tau]$ are the non-signal regions of length $\tau$ on the left and right of the signal regions $I^*$. We formally present these in Conditions (D) and (E) in the following two theorems.
\begin{theorem}\label{th-power-weak}
Assume conditions (A)-(C) in Theorem \ref{th-multi-scan} and the following condition (D) hold,
\begin{itemize}
\item[(D)] $\frac{||\V\mu_{I^*}||_2^2}{||\V\lambda_{I^*}||_2} \geq 2(1+\epsilon_0)\sqrt{\log(p)}$ for some constant $\epsilon_0 > 0$,
\end{itemize}
then,
\begin{equation*}
\P\big\{Q(I^*) > h(p,L_{\min},L_{\max},\alpha)\big\} \rightarrow 1.
\end{equation*}
\end{theorem}
The proof of Theorem 3 is given in the Appendix. Condition (D) imposes on the signal strength of the signal region. This condition is similar to the condition assumed in \cite{jeng2010optimal} and ensures that the signal region $I^*$ will be selected in the candidate set $\mathcal{I}^{(1)}$. For each signal variant $i$, by our definition, $\mu_i = \E(U_i) $ has the same convergence rate as $\sqrt{n}$, where $n$ is the sample size. In GWAS or WGS, the sample size is often large and thus condition (D) is reasonable in reality. Theorem \ref{th-power-weak} could consistently select a signal region that overlaps with the true signal region whose overall signal strength in sense of $L_2$ norm is sufficiently large.

To show the consistency of signal region detection, we first introduce a quantity to measure the accuracy of an estimator of a signal segment. For any two regions $I_1$ and $I_2$, define the Jaccard index between $I_1$ and $I_2$ as $J(I_1,I_2) = |I_1 \cap I_2| / |I_1 \cup I_2|$. It is obvious that $0 \leq J(I_1,I_2) \leq 1$, and $J(I_1,I_2)=1$ indicating complete identity and $J(I_1,I_2)$ indicating disjointness. Let $\hat{\mathcal{I}} = \{ \hat{I}_1,\hat{I}_2,\cdots\}$ be a collection of estimated signal regions, we define region $I^*$ is consistently detected if for some $\eta_p = o(1)$, there exists $\hat{I} \in \hat{\mathcal{I}}$ such that
\begin{equation*}
\P\big(J(\hat{I},I^*) \geq 1 - \eta_p \big) \rightarrow 1.
\end{equation*}
The following theorem shows that the proposed Q-SCAN could consistently detect existence and locations the signal region $I^*$ under some regularity conditions.
\begin{theorem}\label{th-power-strong}
Assume conditions (A)-(D) in Theorem \ref{th-power-weak} and the following condition (E) hold,
\begin{itemize}
\item[(E)] $\inf_{I \subsetneq I^*} \frac{\log(||\V\mu_{I^*}||_2^2) - \log(||\V\mu_I||_2^2)}{\log(||\V\lambda_{I^*}||_2) - \log(||\V\lambda_I||_2)} > 1$,
\end{itemize}
then
\begin{equation*}
\P\big(J(\hat{I},I^*) \geq 1 - \eta_p\big) \rightarrow 1,
\end{equation*}
for any $\eta_p$ that satisfies  $\big\{\frac{\log(L_{\max})}{\log(p)}\big\}^{\frac{1}{4}} \ll \eta_p \ll 1$.
\end{theorem}
The proof of Theorem \ref{th-power-strong} is provided in the Appendix. Condition (E) specifies the properties of the overall signal strength that a signal region needs to satisfy in order for it to be  consistently detected by the Q-SCAN procedure. This definition allows a signal region to consist of both signal and neutral variants, which is more realistic and commonly the case in GWAS and WGS. This condition is implicitly assumed when signals have the same strength and tests are independent. However this common strength assumption that is suitable for copy number variation studies is inappropriate for GWAS and WGS. Condition (E) also holds when the tests are independent and the sparsity parameter is constant in the signal region. To be specific, let $s(I)$ be the number of signals in region $I$, that is, the number of $\mu_i$'s that are not zero in region $I$. Assume $s(I) = p_{I}^{\xi(I)}$, where $\xi(I)=\xi^*$ is the sparsity parameter of region $I$. Although signals are sparse across the genome, we assume that signals are dense in the signal region \citep{donoho2004higher,wu2011rare} and hence $\xi^* > 1/2$. Then, for any $I \subsetneq I^*$, we have $\{\log(||\V\mu_{I^*}||_2^2\} - \log(||\V\mu_I||_2^2))/\{\log(||\V\lambda_{I^*}||_2) - \log(||\V\lambda_I||_2)\} = 2\xi^* > 1$ and thus condition (E) holds.

The results in Theorem \ref{th-power-strong} show that the proposed quadratic scan procedure is consistent for estimating a signal region, and its consistency depends on the signals only through their $L_2$ norm.  This indicates that the direction and  sparsity of the signals in a signal region do not affect  the consistency of the proposed scan procedure. When marginal test statistics are independent and signals have the same strength in the signal region, i.e., $\mu_i=\mu$ for all $i \in I^*$, \cite{jeng2010optimal} developed a theoretically optimal likelihood ratio selection procedure based on the mean scan statistic (\ref{eq:mean scan statistics}).  For the likelihood ratio selection procedure to consistently detect the signal region $I^*$, the condition on $\mu$ is $\mu \geq \sqrt{2(1+\delta_p)\log(p)}/\sqrt{|I^*|}$ for some $\delta_p$ such that $\delta_p\sqrt{\log p} \rightarrow \infty$. It means that $||\V\mu_{I^*}||_2^2 \geq 2(1+\delta_p)\log(p)$. Because $|I^*|/\log(p) \rightarrow \infty$, this condition is weaker than condition (D) in Theorem \ref{th-power-strong}, which is $||\V\mu_{I^*}||_2^2 \geq 2(1+\epsilon_0)\sqrt{|I^*|\log(p)}$ for this situation. However, it is obvious that the quadratic procedure has more power than the mean scan procedure \citep{jeng2010optimal} in the presence of both trait-increasing and trait-decreasing variants in the signal region. The quadratic scan procedure is also more powerful in the presence of weak or neutral variants in the signal region. We will illustrate this in finite sample simulation studies in Section \ref{sec:simulation}.

\section{Simulation Studies}\label{sec:simulation}
\subsection{Family-wise Error Rate for Quadratic Scan Procedure}\label{subsec:fwer control}
In order to validate the proposed quadratic scan procedure in terms of protecting family-wise error rate using the empirical threshold, we estimated the family-wise error rate through simulation. To mimic WGS data, we generated sequence data by simulating 20,000 chromosomes for a 10 Mb region on the basis of the calibration-coalescent model that mimics the linkage disequilibrium structure of samples from African Americans using COSI \citep{sanda2008quality}. The simulation used the 10 Mb sequence to represent the whole genome and focused on low frequency and rare variants with minor allele frequency less than $0.05$. The total sample size $n$ is set to be $2,500$, $5,000$ or $10,000$ and the corresponding number of variants in the sequence are $189,597$, $242,285$ and $302,737$, respectively. We first consider the continuous phenotype generated from the model:
\begin{equation*}
\V Y = 0.5 \V X_1 + 0.5 \V X_2 + \V\epsilon,
\end{equation*}
where $\V X_1$ is a continuous covariate generated from a standard normal distribution, $\V X_2$ is a dichotomous covariate taking values $0$ and $1$ with a probability of $0.5$, and $\V\epsilon$ follows a standard normal distribution. We selected the minimum searching window length $L_{\max}=40$ and the maximum searching window length $L_{\max}=200$. We scan the whole sequence for controlling the family-wise error rate at $0.05$ and $0.01$ level. The simulation was repeated for $10,000$ times.

We also conducted the family-wise error rate simulations for dichotomous phenotypes using similar settings except that the dichotomous outcomes were generated via the model:
\begin{equation*}
\text{logit}\big\{\P(Y_i=1)\big\} = -4.6 + 0.5 \V X_1 + 0.5 \V X_2, ~~i=1,\cdots,n,
\end{equation*}
which means the prevalence is set to be 1\%. Case-control sampling was used and the numbers of cases and controls were equal. The sample sizes were the same as those used for continuous phenotypes.

For both continuous and dichotomous phenotype simulations, we applied the proposed Q-SCAN procedure and M-SCAN procedure to each of the $10,000$ data sets. To control for LD, the mean scan statistic for region $I$ is defined as
\begin{equation}\label{eq:mean_scan_LD}
M(I) = \big(\sum_{i \in I} U_i\big)^2/ var\big(\sum_{i \in I} U_i\big).
\end{equation}

The empirical family-wise type I error rates estimated for Q-SCAN and M-SCAN are presented in Table \ref{tab:fwer} for $0.05$ and $0.01$ levels, respectively. The family-wise error rate is accurate at both two significance levels and all the empirical family-wise error rate fall in the $95\%$ confidence interval of  the 10,000 Bernoulli trials with probability $0.05$ and $0.01$. These results showed that both Q-SCAN procedure and M-SCAN procedure are valid methods and protect the family-wise error rate.

\subsection{Power and Detection Accuracy Comparisons}\label{subsec:power}
In this section, we performed simulation studies to investigate the power of the proposed Q-SCAN procedure in finite samples, and compared its performance with the M-SCAN procedure using scan statistic (\ref{eq:mean_scan_LD}). Genotype data was generated in the same fashion as Section \ref{subsec:fwer control}. The total sample sizes were set as $n=2,500$, $n=5,000$ and $n=10,000$. We generated continuous phenotypes by
\begin{equation*}
\V Y = 0.5 \V X_1 + 0.5 \V X_2 + \V G_1^c\beta_1 + \cdots + \V G_s^c\beta_s + \V\epsilon,
\end{equation*}
where $\V X_1$, $\V X_2$ and $\V\epsilon$ are the same as those specified in the family-wise error rate simulations, $\V G_1^c,\cdots,\V G_{s}^c$ are the genotypes of the $s$ causal variants and $\beta$s are the log odds ratio for the causal variants. We randomly selected two signal regions across the 10 Mb sequence in each replicate and repeated the simulation 1,000 times. The number of variants in each signal region $p_0$ was randomly selected from 50 to 80. Causal variants were selected randomly within each signal region. Assume $\xi$ is the sparsity index, that is, $s = p_0^\xi$. We considered two sparsity settings $\xi=2/3$ and $\xi=1/2$. Each of causal variant has an effect size as a decreasing function of MAF, $\beta=c\log_{10}$(MAF). We set $c=0.185$ for $\xi=2/3$ and $c=0.30$ for $\xi=1/2$. We also considered three settings of effect direction: the sign of $\beta_i$'s are randomly and independently set as 100\% positive (0\% negative), 80\% positive (20\% negative), and 50\% positive (50\% negative).

To evaluate power for these methods, we considered two criteria, the signal region detection rate and the Jaccard index as performance measurements. Let $I_1$ and $I_2$ be the two signal regions and $\{\hat{I}_j\}$ be a collection of signal regions. The signal region detection rate is defined as
\begin{equation*}
\frac{1}{2} \Big [\V 1\big\{I_1\cap(\cup \hat{I}_j) \neq \emptyset \big\} + \V 1\big\{I_2\cap(\cup \hat{I}_j) \neq \emptyset \big\}\Big ],
\end{equation*}
where $\V 1(\cdot)$ is an indicator function. Here we define the signal region as detected if it is overlapped with one of the signal regions. For the Jaccard index, we define it as
\begin{equation*}
\frac{1}{2} \Big\{\max_j J(I_1,\hat{I}_j) + \max_j J(I_2,\hat{I}_j)\Big\},
\end{equation*}
where $J(\cdot,\cdot)$ is defined in Section \ref{subsec:Power analysis}.

\text{Figure \ref{fig:Q_SCAN_2_3}} summarizes the simulation results when the sparsity index is equal to $2/3$. In this situation, the Q-SCAN procedure had a better performance for detecting signal regions than the M-SCAN procedure when the effect of causal variants are in different directions, and was comparable to the M-SCAN procedure when the effects of causal variants are in the same direction. Specifically, when the effects of causal variants are in different directions, the Q-SCAN procedure had higher signal region detection rates and the Jaccard index than the M-SCAN procedure. The difference was more appreciable when the proportion of variants with negative effects increased from 20\% to 50\%. Although the M-SCAN procedure had  higher signal region detection rates and  Jaccard index than the Q-SCAN procedure when the effects of causal variants are in the same direction, the difference decreased as the sample size increased. These results indicated that the performance of the Q-SCAN procedure was robust to the direction of effect sizes, while the M-SCAN procedure loses power when the effect sizes are in different directions.

\text{Figure \ref{fig:Q_SCAN_1_2}} summarizes the simulation results when the sparsity index is equal to $1/2$. In this situation, the Q-SCAN procedure performed better than M-SCAN for detecting signal regions in the sense that Q-SCAN had higher signal region detection rate and Jaccard index. Different from the case that $\xi=2/3$, when signal variants are more sparse in a signal region, the Q-SCAN procedure also has a higher signal region detection rate and Jaccard index than the M-SCAN procedure when the effects of causal variants are in the same direction.

In summary, our simulation study illustrates that Q-SCAN has an advantage in signal region identification over M-SCAN, especially in the presence of signal variants that have effects in different directions or  neutral variants among variants in signal regions. We also did simulations for multiple effect sizes and for sparsity index $\xi=3/4$. The results were similar and could be found in the Supplementary Materials (\text{Figure S1}-\text{Figure S4}).

\section{Application to the ARIC Whole Genome Sequencing (WGS) Data}\label{sec:realdata}
In this section, we analyzed the ARIC WGS study conducted at the Baylor College of Medicine Human Genome Sequencing Center. DNA samples were sequenced at 7.4-fold average depth on Illumina HiSeq instruments. We were interested in detecting genetic regions that were associated with two quantitative traits, small, dense, low-density lipoprotein cholesterol (LDL) and neutrophil count, both of which are risk factors of cardiovascular disease. After sample-level quality control \citep{morrison2017practical}, there are 55 million variants observed in 1,860 African Americans (AAs) and 33 million variants observed in European Americans (EAs). Among these variants, 83\% and 80\% are low-frequency and rare variants (Minor Allele Frequency (MAF) $<$5\%) in AAs and EAs, respectively. For this analysis, we focused on analyzing these low-frequency and rare variants across the whole genome.

To illustrate the proposed Q-SCAN procedure, we compared the performance of the Q-SCAN procedure with the Mean scan procedure M-SCAN and the SKAT \citep{wu2011rare} conducted using a sliding window approach with fixed window sizes. Following \cite{morrison2017practical}, for the sliding window approach, we used the sliding window of length 4 kb and began at position 0 bp for each chromosome, utilizing a skip length of 2 kb, and we tested for the association between variants in each window and the phenotype using SKAT. We adjusted for age, sex, and the first three principal components of ancestry in the analysis for both traits and additionally adjusted for current smoking status in the analysis of neutrophil count, consistent with the procedure described in \cite{morrison2017practical}. Because the distribution of both LDL and neutrophil count are markedly skewed, we transformed them using the rank-based inverse normal-transformation following the standard GWAS practice \citep{Barber}, see (\text{Figure S5} and \text{Figure S6}).

For both Q-SCAN and M-SCAN procedures, we set the range of searching window sizes by specifying the minimum and maximum numbers of variants $L_{\min}=40$ and $L_{\max}=200$. We controlled the family-wise error rate (FWER) at the 0.05 level in both Q-SCAN and M-SCAN analyses using the proposed empirical threshold. For the sliding window procedure, following \cite{morrison2017practical}, we required a minimum number of 3 minor allele counts in a 4 kb window with a skip of length of 2 kb, which results in a total of $1,337,673$ and $1,337,382$ overlapping windows in AA and EA, respectively. As around 1.3 million windows were tested using the sliding window procedure, we used the Bonferroni method to control for the FWER at the $0.05$ level in the sliding window method following the GWAS convention. We hence set the region-based significance threshold for the sliding window procedure at $3.75\times10^{-8}$ (approximately equal to $0.05/1,337,000$). We note that both Q-SCAN and M-SCAN directly control for the FWER without the need of further multiple testing adjustment.

Q-SCAN detected a signal region of $4,501$ basepairs (from $45,382,675$ to $45,387,175$ bp on chromosome 19) consisting of 58 variants that had a significant association with LDL among EAs with the family-wise error rate 0.005.  This region resides in \textit{NECTIN2} and covers three uncommon variants with individual p-values less than $1\times 10^{-6}$, including rs4129120 with $p=8.47\times 10^{-9}$ and MAF$=0.036$, rs283808 with $p=5.71\times 10^{-7}$ and MAF=0.042, and rs283809 with $p=5.71\times 10^{-7}$ and MAF$=0.042$.  Although the variant rs4129120 was significant at level $5\times 10^{-8}$, there are $9,367,575$ variants with MAF$ \geq 0.01$ across the genome and hence the family-wise error rate estimated by Bonferroni correction was 0.079, which was much larger than that of the region detected by Q-SCAN. Several common variants in \textit{NECTIN2} have been found to have a significant association with LDL in previous studies \citep{talmud2009gene,postmus2014pharmacogenetic}.

The M-SCAN procedure did not detect any signal segment associated with LDL to reach genome-wide significance when we controlled for the family-wise error rate at 0.05.  Examination of the data show that the variant effects had different directions and were mixed with neutral variants in the signal region detected by Q-SCAN (the region from $45,382,675$ to $45,387,175$ bp on chromosome 19). The 4 kb sliding window approach using SKAT, which was applied in previous studies \citep{morrison2017practical}, also did not detect any significant windows. Specifically, none of the 4 kb sliding windows covered all of the three variants rs4129120, rs283808 and rs283809, and the SKAT p-values of the two sliding windows that cover variant rs4129120 were $6.6\times 10^{-6}$ and $9.2\times 10^{-7}$, respectively. In contrast, the SKAT p-value of the signal region detected by Q-SCAN was $1.87 \times 10^{-9}$. This indicated that our procedure increased the power for detecting signal regions by estimating the locations of signal regions more accurately. These results explain why our Q-SCAN procedure is more powerful than the M-SCAN procedure and the sliding window procedure using a fixed window size in the analysis of LDL in ARIC WGS data.

We next performed WGS association analysis of neutrophil counts among AAs using Q-SCAN and the existing methods. \text{Figure \ref{fig:Neu-landscape}} summarizes the genetic landscapes of the windows that are significantly associated with neutrophil counts among AAs. All of the significant regions resided in a 6.6 Mb region on chromosome 1. Compared with M-SCAN and the sliding window procedure using SKAT, the significant regions detected by Q-SCAN not only covered the significant regions  detected by these existing methods, but also detected several new regions that were missed by these two methods (Table S1-S3).  Q-SCAN detected 10 novel significant regions, which are separate from the regions detected by  M-SCAN and the sliding window method. For example, the region from 156,116,488 bp to 156,119,540 bp on chromosome 1 showed a significant association with neutrophil counts using Q-SCAN, but was missed by the other two methods. This significant region is 500 kb away from the significant region detected by the other two methods, and resides in gene \textit{SEMA4A}, which has previously shown to have common variants associated with neutrophil related traits \citep{charles2018analyses}.  In summary, compared to the existing methods, the Q-SCAN procedure not only enables the identification of more significant findings, but also is less likely to miss important signal regions.

\section{Discussions}\label{sec:conclusion}
In this paper, we propose a quadratic scan procedure to detect the existence and the locations of signal regions in whole genome array and sequencing studies. We show that the proposed quadratic scan procedure could control for the family-wise error rate using a proper threshold. Under regularity conditions, we also show that our procedure can consistently select the true signal segments and estimate their locations. Our simulation studies  demonstrate that the proposed procedure has a better performance than the mean-based scan method in the presence of variants effects in different directions, or mixed signal variants and neutral variants in signal regions.  An analysis of WGS and heart-and blood-related traits from the ARIC study illustrates the advantages of the proposed Q-SCAN procedure for rare-variant WGS analysis, and demonstrates that Q-SCAN detects the locations and the sizes of signal regions associated with LDL and neutrophil counts  more powerfully and precisely.

The computation time of Q-SCAN is linear in the number of variants in WGS and the range of variants number in searching windows, and hence is computationally efficient. To analyze a 10 Mb region sequenced on 2,500, 5,000 or 10,000 individuals, when we selected the number of variants in searching windows between 40 and 200, the proposed Q-SCAN  procedure took 15, 19 and 25 minutes, respectively, on a 2.90GHz computing cores with 3 Gb memory. Using the same computation core, Q-SCAN took 20 hours to analysis the whole genome of EA individuals from ARIC whole-genome sequencing data. The Q-SCAN procedure also works for parallel computing. Analyzing the whole genome for 10,000 individuals only requires 75 minutes if using 100 computation cores.

We derive an empirical threshold based on Monte Carlo simulations to control for the family-wise error rate at an exact $\alpha-$ level and give the asymptotic bound of the proposed threshold. This step costs additional computation time in applying our procedure. Future research is needed to develop an analytic approximation to the significance level for the proposed Q-SCAN statistics. We allow in this paper individual variant test statistics to be correlated, and assume $M_p-$dependence.  This is a reasonable assumption in WGS association studies, because two marginal test statistic are independent when two variants are far apart in the genome. It is of future research interest to extend our procedure to more general correlation structures.

We assume in this paper all the variants have the same weight in constructing the quadratic scan statistic. In WGS studies, upweighting rare variants and functional variants could boost power when  causal variants are likely to be more rare and functional. It is of great interest to extend our proposed scan procedure by weighting individual test statistics using  external bioinformatic  and evolutionary biology knowledge, such as variant functional information when applying to WGS studies. We assume individual variant test statistics are asymptotically jointly normal. However, when most of variants are rare variants in the sequence, this normal assumption might not hold in finite samples for binary traits. An interesting problem of future research is to extend the results to the situation where individual variant test statistics are not normal and use the exact or approximate  distributions of individual test statistics to construct scan statistics.

\bibliographystyle{apalike}
\bibliography{myrefs}

\renewcommand{\theequation}{A. \arabic{equation}}

\setcounter{equation}{0}
\section*{Appendix}
Due to space limitation, we only present a sketch of the proofs here. The detailed proofs can be found in the  Supplemental Materials. In order to prove the Theorems, we introduce the following three lemmas first.
\begin{lemma}[Refinement of Bernstein's Inequality \citep{petrov1968asymptotic}]\label{lemma:bernstein}
Consider a sequence of independent random variables $\{\V X_j\},j=1,2,\cdots$. Assume $\E(\V X_j) = 0$, $\E(\V X_j^2) = \sigma_j^2$ and $L_j(z) = \log(\E(\exp(z\V X_j)))$. We introduce the following notation:
\begin{equation*}
\V S_n = \sum_{j=1}^n X_j, \quad B_n = \sum_{j=1}^n \sigma_j^2, \quad F_n(x) = \P(\V S_n < x\sqrt{B_n}), \quad \Phi(x) = \frac{1}{\sqrt{2\pi}}\int_{-\infty}^x \exp(-\frac{t^2}{2}) dt. \end{equation*}
If the following three conditions hold:
\begin{itemize}
\item[(1)] There exists positive constants $A,c_1,c_2,\cdots$ such that $|L_j(z)| \leq c_j$ for $|z|<A$, $j=1,2,\cdots$,
\item[(2)] $\limsup_{n\rightarrow \infty} \frac{1}{n}\sum_{j=1}^n c_j^{\frac{3}{2}} < \infty$,
\item[(3)] $\liminf_{n \rightarrow \infty} \frac{B_n}{n} > 0$,
\end{itemize}
then there exists a positive constant $w$ such that for sufficiently large $n$,
\begin{equation*}
1 - F_n(x) = [1-\Phi(x)]\exp\big\{\frac{x^3}{\sqrt{n}}\lambda_n(\frac{x}{\sqrt{n}})\big\}(1+l_1w)
\end{equation*}
\begin{equation*}
F_n(-x) = \Phi(-x) \exp\{-\frac{x^3}{\sqrt{n}}\lambda_n(-\frac{x}{\sqrt{n}})(1+l_2w)\}
\end{equation*}
in the region $0 \leq x \leq w\sqrt{n}$. Here $|l_1|\leq l$ and $|l_2| \leq l$ being some constant.
\end{lemma}
The next two lemmas are two inequalities for the tail probability of normal distribution and chi-square distribution, respectively.
\begin{lemma}[Mills' Ratio Inequality]\label{lemma:mills}
For arbitrary positive number $x>0$, the inequalities
\begin{equation*}
\frac{x}{1+x^2} \exp(-\frac{x^2}{2}) < \int_x^{\infty} \exp(-\frac{u^2}{2}) du < \frac{1}{x}\exp(-\frac{x^2}{2}).
\end{equation*}
\end{lemma}

\begin{lemma}[Exponential inequality for chi-square distribution\citep{laurent2000adaptive}]\label{lemma:chisq}
Let $\V Y_1,\cdots, \V Y_D$ be i.i.d. standard Gaussian variables. Let $a_1,\cdots,a_D$ be nonnegative constants. Let $\V Z = \sum_{i=1}^D a_i(\V Y_i^2 - 1)$, then the following inequalities hold for any positive x:
\begin{equation*}
\P(\V Z \geq 2||\V a||_2 \sqrt{x} + 2||\V a||_{\infty}x) \leq \exp(-x),
\end{equation*}
\begin{equation*}
\P(\V Z \leq -2||\V a||_2 \sqrt{x} \leq \exp(-x).
\end{equation*}
\end{lemma}

\subsection*{Bound of Empirical Threshold}
Now we prove the main theorem in the paper. We first show the results of a fixed length scan.
\begin{lemma}\label{th-fixed-scan}
If the following conditions (A)-(C) hold,
\begin{itemize}
\item[(A)] $\max_{|I|=L_p} ||\V\lambda_I||_{\infty} \leq K_0$, where $K_0$ is a constant,
\item[(B)] $\frac{L_p}{\log p} \rightarrow \infty$ and $\frac{\log(L_p)}{\log(p)} \rightarrow 0$,
\item[(C)] $\{U_i\}_{i=1}^p$ is $M_p$-dependent and $\frac{\log(M_p)}{\log(p)} \rightarrow 0$,
\end{itemize}
then
\begin{equation*}
\frac{\max_{|I|=L_p} Q(I)}{\sqrt{2\log(p)}} \stackrel{p}{\rightarrow} 1.
\end{equation*}
\end{lemma}

\begin{proof}
For any $I$ that satisfies that $|I|=L_p$, assume $\V \Omega_I \V\Sigma_I \V \Omega_I^T = \diag(\V\lambda_I) $ is the SVD of $\V\Sigma_I$ where $\V\Omega_I$ is an orthogonal matrix. Let $\tilde{\V U}_I = \V\Omega_I\V U_I$, then $\tilde{\V U}_I \sim N\big(0,\diag(\V\lambda_I)\big)$ and $\tilde{\V U}_I^T\tilde{\V U}_I = \V U_I^T\V U_I$. Hence, for $\epsilon>0$, $ \P(Q(I)>(1+\epsilon)\sqrt{2\log(p)}) = \P\big(\sum_{i \in I}\lambda_{I,i}(\frac{\tilde{U}_i^2}{\lambda_{I,i}}-1)>2(1+\epsilon)||\V\lambda_I||_2\sqrt{\log(p)}\big).$ Note that $||\V\lambda_I||_2 \geq ||\V\lambda_I||_1^2/\sqrt{|I|} \geq c\sqrt{|I|}$, by conditions (A) and (B), $ (1+\frac{\epsilon}{2})^2||\V\lambda_I||_{\infty}\log(p) = O\big(\log(p)\big) = o(||\V\lambda_I||_2\sqrt{\log p}).$ Then by Lemma \ref{lemma:chisq}, for $p$ sufficiently large,
\begin{equation*}
\P\Big(\sum_{i \in I} \lambda_{I,i}(\frac{\tilde{U}_i^2}{\lambda_{I,i}}-1) > 2(1+\epsilon)||\V\lambda_I||_2\sqrt{\log(p)}\Big) \leq \exp\{-(1+\frac{\epsilon}{2})^2\log(p)\}.
\end{equation*}
Hence, using Boole's inequality,
\begin{equation}\label{th1:p1}
\P\Big(\max_{|I|=L_p} Q(I) > (1+\epsilon)\sqrt{2\log(p)}\Big) \leq \sum_{|I|=L_p}\P\Big(Q(I) > (1+\epsilon)\sqrt{2\log(p)}\Big)  = o(1).
\end{equation}
On the other hand, $ \P\Big(\max_{|I|=L_p} Q(I) < (1-\epsilon)\sqrt{2\log(p)}\Big) = 1 - \P\big(\bigcup_{|I|=L_p}\{Q(I) \geq (1-\epsilon)\sqrt{2\log(p)}\}\big)$. Let $I_k=\{k,k+1\cdots,k+L_p-1\}$ and $A_k = \Big\{\frac{\sum_{i \in I_k} U_i^2 - ||\V\lambda_{I_k}||_1}{\sqrt{2}||\V\lambda_{I_k}||_2} \geq (1-\epsilon)\sqrt{2\log(p)}\Big\}$, by Chung-Erd\"{o}s inequality, $ \P\big(\bigcup_{|I|=L_p}\{Q(I) \geq (1-\epsilon)\sqrt{2\log(p)}\}\big) \geq \frac{\{\sum_{k=1}^{p-L_p+1}\P(A_k)\}^2}{\sum_{i=1}^{p-L_p+1}\sum_{j=1}^{p-L_p+1}\P(A_i \cap A_j)}$. Note that
\begin{eqnarray*}
\sum_{i=1}^{p-L_p+1}\sum_{j=1}^{p-L_p+1}\P(A_i \cap A_j)
&=& \sum_{i=1}^{p-L_p+1} \P(A_i) + \sum_{|i-j|\leq M_p+L_p} \P(A_i \cap A_j) + \sum_{|i-j| > M_p+L_p} \P(A_i)\P(A_j)
\\&\leq& (M_p+L_p+1) \sum_{i=1}^{p-L_p+1} \P(A_i) + \sum_{|i-j|>M_p+L_p} \P(A_i)\P(A_j),
\end{eqnarray*}
then we get $ \frac{\{\sum_{k=1}^{p-L_p+1}\P(A_k)\}^2}{\sum_{i=1}^{p-L_p+1}\sum_{j=1}^{p-L_p+1}\P(A_i \cap A_j)} \geq 1 - \frac{M_p+L_p+1}{\sum_{k=1}^{p-L_p+1}\P(A_k)}$. By conditions (A) and (B), using Lemma \ref{lemma:bernstein}, there exist constants $c_1,c_2>0$ such that
\begin{equation*}
\sum_{k=1}^{p-L_p+1} \P(A_k) \geq \P\big[1-\Phi\{(1-\epsilon)\sqrt{2\log(p)}\}\big]\exp\big\{\frac{2^{\frac{3}{2}}(1-\epsilon)^3(\log(p))^{\frac{3}{2}}}{\sqrt{L}}c_1\big\}c_2.
\end{equation*}
Then, using Lemma \ref{lemma:mills}, we have $\sum_{k=1}^{p-L_p+1} \P(A_k) \geq \frac{c_2}{4(1-\epsilon)}\exp\big\{\log(p)(2\epsilon-\epsilon^2+2^{\frac{3}{2}}c_1(1-\epsilon)^3\sqrt{\frac{\log(p)}{L_p}})-\frac{1}{2}\log(\log(p))\big\}.$ Hence by conditions (B) and (C), $ \sum_{k=1}^{p-L_p+1}\P(A_k) / (M_p+L_p+1) \rightarrow \infty.$
This indicates that
\begin{equation}\label{th1:p2}
\P\big(\max_{|I|=L_p} Q(I) < (1-\epsilon)\sqrt{2\log(p)}\big) \rightarrow 0.
\end{equation}
Combining (\ref{th1:p1}) and (\ref{th1:p2}), we get $ \P\Big(|\frac{\max_{|I|=L_p} Q(I)}{\sqrt{2\log(p)}} - 1| > \epsilon \Big) \rightarrow 0.$ As  $\epsilon$ is arbitrary, we complete the proof. \QEDB
\end{proof}
We then prove the result of multi-length scan, which is the Theorem 1 in the paper.
\begin{proof}
For any $\epsilon > 0$,
\begin{eqnarray*}
&&\P\Big(\big|\frac{\max_{L_{\min} \leq |I| \leq L_{\max}} Q(I)}{\sqrt{2\log(p)}} - 1 \big| \geq \epsilon\Big)
\\&=& \P\Big(\max_{L_{\min} \leq |I| \leq L_{\max}} Q(I) \geq (1+\epsilon)\sqrt{2\log(p)}\Big) + \P\Big(\max_{L_{\min} \leq |I| \leq L_{\max}} Q(I) \leq (1-\epsilon)\sqrt{2\log(p)}\Big)
\\&\triangleq& A_1(\epsilon) + A_2(\epsilon).
\end{eqnarray*}
Note that $\max_{L_{\min} \leq |I| \leq L_{\max}} Q(I) \geq \max_{|I|=L_{\max}} Q(I)$, by Lemma \ref{th-fixed-scan},
\begin{equation}\label{th2:p1}
A_2(\epsilon) \leq \P\Big(\max_{|I|=L_{\max}} Q(I) \leq (1-\epsilon)\sqrt{2\log(p)}\Big) = o(1).
\end{equation}
For $A_1(\epsilon)$, by Boole's inequality, $ A_1(\epsilon) \leq \sum_{L_{\min} \leq |I| \leq L_{\max}} \P\Big(Q(I) \geq (1+\epsilon)\sqrt{2\log(p)}\Big).$ For any $I$ satisfies that $L_{\min} \leq |I| \leq L_{\max}$, by condition (B), $ \frac{2(1+\frac{\epsilon}{2})^2\log(p)||\V\lambda_I||_{\infty}}{||\V\lambda_I||_2} \leq \frac{2K_0(1+\frac{\epsilon}{2})^2\log(p)}{\sqrt{L_{\min}}} = o(\sqrt{\log(p)}).$ Hence, by Lemma \ref{lemma:chisq},
\begin{equation}\label{th2:p2}
A_1(\epsilon) \leq \sum_{L_{\min} \leq |I| \leq L_{\max}} \exp\{-(1+\frac{\epsilon}{2})^2\log(p)\}
\leq \exp\{-(1+\frac{\epsilon}{2})^2\log(p) + \log(p) + \log(L_{\max})\}
= o(1).
\end{equation}
Using (\ref{th2:p1}) and (\ref{th2:p2}), we have $ \P\Big(\big|\frac{\max_{L_{\min} \leq |I| \leq L_{\max}} Q(I)}{\sqrt{2\log(p)}} - 1 \big| \geq \epsilon\Big) = o(1).$ By the arbitrary of $\epsilon$, we complete the proof. \QEDB
\end{proof}
Next we prove Proposition \ref{prop:thres-upper}, which gives an upper bound of $h(p,L_{\min},L_{\max},\alpha)$.
\begin{proof}
Let $\tilde{h}(p,\alpha) = \sqrt{2\big[\log\{p(L_{\max}-L_{\min})\} - \log(\alpha)\big]} + \frac{\sqrt{2}\big[\log\{p(L_{\max}-L_{\min})\} - \log(\alpha)\big]}{\{L_{\min}\log(p)\}^{\frac{1}{4}}}$. For any $I$ that satisfies  $L_{\min} \leq |I| \leq L_{\max}$, we have
\begin{equation*}
\frac{2||\lambda_I||_{\infty}\big[\log\{p(L_{\max}-L_{\min})\} - \log(\alpha)\big]}{\sqrt{2}||\V\lambda_I||_2}
= O\Big(\frac{\sqrt{2}\big[\log\{p(L_{\max}-L_{\min})\} - \log(\alpha)\big]}{\sqrt{|I|}}\Big).
\end{equation*}
Note that $\frac{\log(p)}{|I|} \leq \frac{\log(p)}{L_{\min}} = o(1)$, then for sufficiently large $p$,
\begin{equation*}
\frac{2||\V\lambda_I||_{\infty}\big[\log\{p(L_{\max}-L_{\min})\} - \log(\alpha)\big]}{\sqrt{2}||\V\lambda_I||_2} \leq \frac{\sqrt{2}\big[\log\{p(L_{\max}-L_{\min})\} - \log(\alpha)\big]}{\big(L_{\min}\log(p)\big)^{\frac{1}{4}}}.
\end{equation*}
Hence, by Lemma \ref{lemma:chisq},
$\P\big(Q_{\max} > \tilde{h}(p,\alpha) \big) \leq \sum_{L_{\min} \leq |I| \leq L_{\max}} \exp\big(-\log\{p(L_{\max}-L_{\min})\} - \log(\alpha)\big) \leq \alpha.$ Thus $h(p,L_{\min},L_{\max},\alpha) \leq \tilde{h}(p,\alpha)$ and we complete the proof.
\QEDB
\end{proof}

\subsection*{Consistency of Signal Region Detection}
In this section, we show the results of power analysis. We first prove Theorem \ref{th-power-weak}, which shows that the proposed Q-SCAN procedure could consistently select a signal region that overlaps with the true signal region.
\begin{proof}
By the definition of $Q(I)$,
$\P\big(Q(I^*) \leq h(p,L_{\min},L_{\max},\alpha)\big)
\leq \P\Big(\big|\sum_{i \in I^*} U_i^2 - ||\V\lambda_{I^*}||_1 - ||\V\mu_{I^*}||_2^2\big| \geq ||\V\mu_{I^*}||_2^2 - \sqrt{2}||\V\lambda_{I^*}||_2h(p,L_{\min},L_{\max},\alpha)\Big).$ Because $\Var(\sum_{i \in I^*} U_i^2)  \leq 4||\V\lambda_{I^*}||_{\infty}||\V\mu_{I^*}||_2^2 + 2 ||\V\lambda_{I^*}||_2^2$, by using condition (D) and Theorem \ref{th-multi-scan}, for $p$ sufficiently large, $ \frac{\sqrt{2}||\V\lambda_{I^*}||_2h(p,L_{\min},L_{\max},\alpha)}{||\V\mu_{I^*}||_2^2} \leq \frac{h(p,L_{\min},L_{\max},\alpha)}{\sqrt{2}(1+\epsilon)\sqrt{\log(p)}} \leq 1 - \frac{\epsilon_0}{2}.$ Hence, by Chebyshev's inequality, $ \P\big(Q(I^*) \leq h(p,L_{\min},L_{\max},\alpha) \big) \leq \frac{16||\V\lambda_{I^*}||_{\infty}||\V\mu_{I^*}||_2^2 + 8||\V\lambda_{I^*}||_2^2}{\epsilon_0^2||\V\mu_{I^*}||_2^4} = o(1).$ \QEDB
\end{proof}
Next we prove Theorem \ref{th-power-strong}, which shows that the proposed Q-SCAN could consistently detect the existence and location of the true signal region.
\begin{proof}
By the definition of $J(\hat{I},I^*)$,
\begin{eqnarray*}
\P\big(J(\hat{I},I^*) \geq 1 - \eta_p\big)
&\leq& \P\big(\sup_{I \cap I^* \neq \emptyset} Q(I) > (1+\frac{\epsilon_0}{4})\sqrt{2\log(p)}\big) + \P\big(Q(I^*) \leq (1+\frac{\epsilon_0}{4})\sqrt{2\log(p)}\big)
\\&& + \P\big(J(\hat{I},I^*) < 1 - \eta_p, ~\hat{I} \cap I^* \neq \emptyset\big).
\end{eqnarray*}
For the first part, by Theorem \ref{th-multi-scan}, $ \P\big(\sup_{I \cap I^* \neq \emptyset} Q(I) > (1+\frac{\epsilon}{4})\sqrt{2\log(p)}\big) = o(1).$ For the second part, by the same approach discussed in Theorem \ref{th-power-weak}, we have $ \P\big(Q(I^*) \leq (1+\frac{\epsilon_0}{4})\sqrt{2\log(p)}\big) = o(1).$ We next consider the third part. Note that $J(\hat{I},I^*) < 1- \eta_p$ implies that $|\hat{I}\backslash I^*| > \frac{\eta_p}{2}|I^*|$ or $|I^* \backslash \hat{I}| > \frac{\eta_p}{2}|I^*|$,
\begin{equation}\label{th5:p3}
\P\big(J(\hat{I},I^*) < 1 - \eta_p, ~\hat{I} \cap I^* \neq \emptyset\big)
\leq \sum_{\substack{I \cap I^* \neq \emptyset\\|I\backslash I^*|>\frac{\eta_p}{2}|I^*|}} \P\big(Q(I) - Q(I^*) > 0\big) + \sum_{\substack{I \cap I^* \neq \emptyset \\ |I^*\backslash I|>\frac{\eta_p}{2}|I^*|}}\P\big(Q(I) - Q(I^*) > 0\big).
\end{equation}
Let $\tilde{U}_i = U_i - \mu_i$, for any $I \cap I^* \neq \emptyset$, we have
\begin{eqnarray}\label{th5:Q-p}
Q(I) - Q(I^*)
&=& \frac{\sum_{i \in I \backslash I^*} \tilde{U}_i^2 - ||\V\lambda_{I \backslash I^*}||_1}{\sqrt{2}||\V\lambda_I||_2} + (1-\frac{||\V\lambda_I||_2}{||\V\lambda_{I^*}||_2}) \frac{\sum_{i \in I \cap I^*}\tilde{U}_i^2 - ||\V\lambda_{I \cap I^*}||_1}{\sqrt{2}||\V\lambda_I||_2} \nonumber
\\&& + \frac{2\sum_{i \in I \cap I^*} \mu_i\tilde{U}_i}{\sqrt{2}||\V\lambda_I||_2}(1-\frac{||\V\lambda_I||_2}{||\V\lambda_{I^*}||_2}) - \frac{2\sum_{i \in I^* \backslash I} \mu_i\tilde{U}_i}{\sqrt{2}||\V\lambda_{I^*}||_2} - \frac{\sum_{i \in I^* \backslash I}\tilde{U}_i^2 - ||\V\lambda_{I^* \backslash I}||_1}{\sqrt{2}||\V\lambda_{I^*}||_2} \nonumber
\\&& + \frac{||\V\mu_{I \cap I^*}||_2^2}{\sqrt{2}||\V\lambda_I||_2} - \frac{||\V\mu_{I^*}||^2_2}{\sqrt{2}||\V\lambda_{I^*}||_2}.
\end{eqnarray}
Assume $\delta = \inf_{I \subsetneq I^*} \frac{\log(||\V\mu_{I^*}||_2^2) - \log(||\V\mu_I||_2^2)}{\log(||\V\lambda_{I^*}||_2) - \log(||\V\lambda_I||_2)} - 1$, then by condition (E), there exists $\delta > 0$ and for any $I$, $ \frac{||\V\mu_{I^*}||_2^2}{||\V\lambda_{I^*}||_2} - \frac{||\V\mu_{I \cap I^*}||_2^2}{||\V\lambda_{I}||_2} \geq \frac{||\V\mu_{I^*}||_2^2}{||\V\lambda_{I^*}||_2}\big(1 - \frac{||\V\lambda_{I \cap I^*}||_2^{1+\delta}}{||\V\lambda_I||_2||\V\lambda_{I^*}||_2^{\delta}}\big).$ By (\ref{th5:Q-p}), it follows that $ \P\big(Q(I) - Q(I^*) > 0\big) \leq \sum_{i=1}^5 \P\big(A_i(I)\big)$,
where $A_1(I) = \Big\{\frac{\sum_{i \in I \backslash I^*} \tilde{U}_i^2 - ||\V\lambda_{I \backslash I^*}||_1}{\sqrt{2}||\V\lambda_I||_2} > \frac{||\V\mu_{I^*}||_2^2}{5\sqrt{2}||\V\lambda_{I^*}||_2}\big(1-\frac{||\V\lambda_{I \cap I^*}||_2^{1+\delta}}{||\V\lambda_I||_2||\V\lambda_{I^*}||_2^{\delta}}\big)\Big\}$,
$A_2(I) = \Big\{(1-\frac{||\V\lambda_I||_2}{||\V\lambda_{I^*}||_2}) \frac{\sum_{i \in I \cap I^*}\tilde{U}_i^2 - ||\V\lambda_{I \cap I^*}||_1}{\sqrt{2}||\V\lambda_I||_2} >  \frac{||\V\mu_{I^*}||_2^2}{5\sqrt{2}||\V\lambda_{I^*}||_2}\big(1-\frac{||\V\lambda_{I \cap I^*}||_2^{1+\delta}}{||\V\lambda_I||_2||\V\lambda_{I^*}||_2^{\delta}}\big)\Big\}$,
$A_3(I) = \Big\{- \frac{\sum_{i \in I^* \backslash I}\tilde{U}_i^2 - ||\V\lambda_{I^* \backslash I}||_1}{\sqrt{2}||\V\lambda_{I^*}||_2} > \frac{||\V\mu_{I^*}||_2^2}{5\sqrt{2}||\V\lambda_{I^*}||_2}\big(1-\frac{||\V\lambda_{I \cap I^*}||_2^{1+\delta}}{||\V\lambda_I||_2||\V\lambda_{I^*}||_2^{\delta}}\big)\Big\}$,
$A_4(I) = \Big\{ \frac{2\sum_{i \in I \cap I^*} \mu_i\tilde{U}_i}{\sqrt{2}||\V\lambda_I||_2}(1-\frac{||\V\lambda_I||_2}{||\V\lambda_{I^*}||_2})  > \frac{||\V\mu_{I^*}||_2^2}{5\sqrt{2}||\V\lambda_{I^*}||_2}\big(1-\frac{||\V\lambda_{I \cap I^*}||_2^{1+\delta}}{||\V\lambda_I||_2||\V\lambda_{I^*}||_2^{\delta}}\big)\Big\}$,
and $A_5(I) = \Big\{- \frac{2\sum_{i \in I^* \backslash I} \mu_i\tilde{U}_i}{\sqrt{2}||\V\lambda_{I^*}||_2} > \frac{||\V\mu_{I^*}||_2^2}{5\sqrt{2}||\V\lambda_{I^*}||_2}\big(1-\frac{||\V\lambda_{I \cap I^*}||_2^{1+\delta}}{||\V\lambda_I||_2||\V\lambda_{I^*}||_2^{\delta}}\big)\Big\}$.

We first consider $A_1(I)$. When $|I \backslash I^*| \geq \frac{\eta_p|I^*|}{2}$, it is obvious that $\frac{|I \backslash I^*|}{|I|} = \frac{|I \backslash I^*|}{|I \backslash I^*| + |I \cap I^*|} \geq \frac{\eta_p}{3} $, then we have $ \frac{||\V\lambda_I||_2^2 - ||\V\lambda_{I \cap I^*}||_2^2}{||\V\lambda_I||_2^2} \geq \frac{c^2\eta_p}{3K_0^2}$
and $\frac{||\V\lambda_I||_2^2 - ||\V\lambda_{I \cap I^*}||_2^2}{||\V\lambda_{I^*}||_2^2} \geq \frac{c^2\eta_p}{2K_0^2}$. It follows that $ \frac{||\V\lambda_I||_2^2 - ||\V\lambda_{I \cap I^*}||_2^2}{||\V\lambda_I||_2(||\V\lambda_{I^*}||_2+||\V\lambda_I||_2)} \geq \frac{2c^2\eta_p}{9K_0^2}$, and hence
$ \frac{||\V\mu_{I^*}||_2^2(||\V\lambda_I||_2||\V\lambda_{I^*}||_2^{\delta}-||\V\lambda_{I \cap I^*}||_2^{1+\delta})}{||\V\lambda_{I^*}||_2^{1+\delta}} \geq \frac{2c^2\eta_p||\V\mu_{I^*}||_2^2}{9K_0^2}.$
On the other hand,
$\frac{||\V\lambda_I||_2^2 - ||\V\lambda_{I \cap I^*}||_2^2}{(||\V\lambda_I||_2 + ||\V\lambda_{I^*}||_2)||\V\lambda_{I^*}||_2} \geq \frac{c\eta_p||\V\lambda_{I^* \backslash I}||_2}{4K_0||\V\lambda_{I^*}||_2}$, and thus
$\frac{||\V\mu_{I^*}||_2^2(||\V\lambda_I||_2||\V\lambda_{I^*}||_2^{\delta}-||\V\lambda_{I \cap I^*}||_2^{1+\delta})}{||\V\lambda_{I^*}||_2^{1+\delta}}
\geq \frac{||\V\mu_{I^*}||_2^2c\eta_p||\V\lambda_{I^* \backslash I}}{4K_0||\V\lambda_{I^*}||_2}$. Note that $6||\V\lambda_{I\backslash I^*}||_{\infty}\log(L_{\max})= o(\eta_p||\V\mu_{I^*}||_2^2)$ and $ 2||\V\lambda_{I\backslash I^*}||_2\sqrt{3\log(L_{\max})} = o\big(||\V\lambda_{I^* \backslash I}||_2\eta_p \sqrt{\log(p)}\big)$, by Lemma \ref{lemma:chisq}, we have $\P\big(A_1(I)\big) \leq 1/L_{\max}^3$ and thus
\begin{equation}\label{th5:p3-A1-1}
\sum_{\substack{I \cap I^* \neq \emptyset \\ |I \backslash I^*| > \frac{\eta_p}{2}|I^*|}} \P\big(A_1(I)\big) \leq \frac{2L_{\max}^2}{L_{\max}^3} = \frac{2}{L_{\max}} = o(1).
\end{equation}
When $|I^* \backslash I| \geq \frac{\eta_p}{2}|I^*|$, we have $ \frac{||\V\lambda_{I \cap I^*}||_2^2}{||\V\lambda_{I^*}||_2^2}  \leq 1 - \frac{c^2\eta_p}{2K_0^2}$, and then
$1 - \Big(\frac{||\V\lambda_{I\cap I^*}||_2}{||\V\lambda_{I^*}||_2}\Big)^{\delta} \geq \frac{\delta c^2\eta_p}{4K_0^2}.$
Note that $\frac{||\V\mu_{I^*}||_2^2(||\V\lambda_I||_2||\V\lambda_{I^*}||_2^{\delta} - ||\V\lambda_{I\cap I^*}||_2^{1+\delta})}{||\V\lambda_{I^*}||_2^{1+\delta}}
\geq \frac{\delta c^2 \eta_p\sqrt{\log(p)}||\V\lambda_I||_2}{K_0^2}$, and $||\V\lambda_I||_2\sqrt{\log p} \geq c\sqrt{|I|\log(p)} \geq c\sqrt{L_{\min}\log(p)}$, we have $ 2||\V\lambda_{I \backslash I^*}||_2\sqrt{3\log(L_{\max})} + 6 ||\V\lambda_{I \backslash I^*}||_{\infty}\log(L_{\max}) = o(\eta_p \sqrt{\log(p)}||\V\lambda_I||_2).$ Then by Lemma \ref{lemma:chisq}, we have $\P\big(A_1(I)\big) \leq 1/L_{\max}^3$ and thus
\begin{equation}\label{th5:p3-A1-2}
\sum_{\substack{I^* \cap I \neq \emptyset \\ |I \backslash I^*| > \frac{\eta_p}{2}|I^*|}} \P\big(A_1(I)\big) \leq \frac{2L_{\max}^2}{L_{\max}^3} = \frac{2}{L_{\max}} = o(1).
\end{equation}

By the same approach we used for $A_1(I)$, for $A_2(I)$ and $A_3(I)$, we could get
\begin{equation}\label{th5:p3-A2-A3}
\sum_{I \cap I^* \neq \emptyset} \P\big(A_2(I)\big)  = o(1), \sum_{\substack{I\cap I^* \neq \emptyset \\ |I^* \backslash I| \geq \frac{\eta_p}{2}|I^*|}} \P\big(A_3(I)\big) = o(1).
\end{equation}

We now consider $A_4(I)$. Let $ \kappa = \frac{\frac{||\V\mu_{I^*}||_2^2(||\V\lambda_I||_2||\V\lambda_{I^*}||_2^{\delta} - ||\V\lambda_{I \cap I^*}||_2^{1+\delta})}{||\V\lambda_I||_2||\V\lambda_{I^*}||_2^{1+\delta}}}{\frac{2\sqrt{\V\mu_{I \cap I^*}^T\V\Sigma_{I \cap I^*}\V\mu_{I \cap I^*}}\big| ||\V\lambda_{I^*}||_2 - ||\V\lambda_I||_2 \big|}{||\V\lambda_I||_2||\V\lambda_{I^*}||_2}}$. By condition (D) and (E), we have $\kappa \geq \frac{\big(\log(p)\big)^{\frac{1}{4}}||\V\lambda_{I^*}||_2^{1-\frac{\delta}{2}}(||\V\lambda_I||_2||\V\lambda_{I^*}||_2^{\delta} - ||\V\lambda_{I \cap I^*}||_2^{1+\delta})}{\sqrt{K_0}\big|||\lambda_{I^*}||_2 - ||\V\lambda_I||_2\big|\times ||\V\lambda_{I\cap I^*}||_2^{\frac{1+\delta}{2}}} $. When $||\V\lambda_I||_2 \geq ||\V\lambda_{I^*}||_2$, $\kappa \geq \frac{\sqrt{\log(p)}}{\sqrt{K_0}}$. Note that $ \Var\Big\{\frac{\sqrt{2}\sum_{i \in I \cap I^*}\mu_i\tilde{U}_i}{||\V\lambda_I||_2}\big(1-\frac{||\V\lambda_I||_2}{||\V\lambda_{I^*}||_2}\big)\Big\}
= \frac{2\V\mu_{I \cap I^*}\V\Sigma_{I \cap I^*}\V\mu_{I \cap I^*}(||\V\lambda_{I^*}||_2 - ||\V\lambda_I||_2)^2}{||\V\lambda_I||_2^2||\V\lambda_{I^*}||_2^2} $,
by Lemma \ref{lemma:mills}, $\P\big(A_4(I)\big)  \leq \frac{1}{L_{\max}^2\log(p)}$. When $||\V\lambda_{I^*}||_2 \geq ||\V\lambda_{I}||_2$, $\kappa \geq \frac{\delta\big(\log(p)\big)^{\frac{1}{4}}||\V\lambda_I||_2^{\frac{1}{2}-\frac{\delta}{2}}||\V\lambda_{I^*}||_2^{\frac{\delta}{2}}}{\sqrt{K_0}}$. Hence, by Lemma \ref{lemma:mills}, $\P\big(A_4(I)\big) \leq \frac{1}{|I^*|^2\log(p)}$ . Therefore, note that $||\V\lambda_{I^*}||_2 \geq ||\V\lambda_I||_2$ implies that $|I| \leq K_0^2|I^*|/c^2$,
\begin{equation}\label{th5:p3-A4}
\sum_{I \cap I^* \neq \emptyset} \P\big(A_4(I)\big)
= \sum_{\substack{I \cap I^* \neq \emptyset \\ ||\V\lambda_I||_2 \geq ||\V\lambda_{I^*}||_2}} \P\big(A_4(I)\big) + \sum_{\substack{I \cap I^* \neq \emptyset \\ ||\V\lambda_I||_2 < ||\V\lambda_{I^*}||_2}} \P\big(A_4(I)\big)
 = o(1).
\end{equation}
Finally, we consider $A_5(I)$, by a similar approach to that used for $A_4(I)$, we have
\begin{equation}\label{th5:p3-A5-2}
\sum_{\substack{I \cap I^* \neq \emptyset \\ |I \backslash I^*| > \frac{\eta_p}{2}|I^*|}} \P\big(A_5(I)\big) = o(1).
\end{equation}
Therefore, by (\ref{th5:p3}) and (\ref{th5:p3-A1-1})-(\ref{th5:p3-A5-2}), we have $ \P\big(J(\hat{I},I^*) < 1 - \eta_p, ~\hat{I} \cap I^* \neq \emptyset\big) = o(1).$
\QEDB
\end{proof}

\pagebreak
\newpage
\clearpage
\thispagestyle{empty}

\begin{table}[H]
\centering
\caption{Simulation Studies of Family-Wise Error Rates. The family-wise error rate of the Q-SCAN procedure is estimated with $10000$ simulated data sets. In each data set, we considered three sample sizes $n=2,500$, $n=5,000$ or $n=10,000$, and the corresponding numbers of variants in the sequence are 189,597, 242,285 and 302,737. The minimum and maximum searching window lengths are set to be $L_{\min} = 40$ and $L_{\max} = 200$, respectively. Q-SCAN refers to the scan procedures using the scan statistics  $Q(I) = \sum_{i \in I} U_i^2  - \E(\sum_{i \in I} U_i^2)/var(\sum_{i \in I} U_i^2)$ for region $I$, where $U_i$ is the score statistic of $i$th variant. M-SCAN refers to the scan procedures using the scan statistics $M(I) = \big(\sum_{i \in I} U_i\big)^2/ var\big(\sum_{i \in I} U_i\big)$ for region $I$. The $95\%$ confidence interval of $10,000$ Bernoulli trials with probability $0.05$ and $0.01$ are $[0.0457,0.0543]$ and $[0.0080,0.0120]$.}
\label{tab:fwer}
\medskip
\begin{tabular}{cccccccccc}
\toprule
   &     & \multicolumn{3}{c}{Continuous Phenotypes} & \multicolumn{3}{c}{Dichotomous Phenotypes} \\
        \cmidrule(r){3-5} \cmidrule(r){6-8}
Total Sample Size  & Size       & $n=2500$ & $n=5000$  & $n=10000$ & $n=2500$   & $n=5000$  & $n=10000$   \\
\midrule
Q-SCAN & 0.05 & 0.0485 & 0.0475  & 0.0514 & 0.0414 &  0.0481  &  0.0488  \\
       & 0.01 & 0.0098 & 0.0100  & 0.0101 & 0.0077 &  0.0104  &  0.0092 \\
M-SCAN & 0.05 & 0.0453 & 0.0525  & 0.0494 & 0.0485 &  0.0501  &  0.0492  \\
       & 0.01 & 0.0097 & 0.0123  & 0.0083 & 0.0094 &  0.0085  &  0.0100  \\
\bottomrule
\end{tabular}
\end{table}

\pagebreak
\newpage
\thispagestyle{empty}

\begin{figure}[H]
\caption{Power and accuracy of estimated signal region comparisons using Q-SCAN and M-SCAN for the setting with the sparsity index $\xi=2/3$. We evaluated power via the signal region detection rate and the Jaccard index defined in the simulation section. Both criteria were calculated at the family-wise error rate at 0.05 level. The number of variants in signal region $p_0$ was randomly selected between 50 and 80.  and $s = p_0^\xi$ causal variants were selected randomly within each signal region. Each  causal variant has an effect size that is set as a decreasing function of MAF, $\beta=c\log_{10}$(MAF) and $c=0.185$.  From left to right, the plots consider settings in which the coefficients for the causal rare variants are 100\% positive (0\% negative), 80\% positive (20\% negative), and 50\% positive (50\% negative). We repeated the simulation for 1000 times. Q-SCAN and M-SCAN refer to the scan procedures using the scan statistics $\sum_{i \in I} U_i^2  - \E(\sum_{i \in I} U_i^2)/var(\sum_{i \in I} U_i^2)$ and $\big(\sum_{i \in I} U_i\big)^2/ var\big(\sum_{i \in I} U_i\big)$, respectively. In both two scan procedures, the range of the numbers of variants in searching windows were between 40 and 200.}
\centering
\label{fig:Q_SCAN_2_3}
\includegraphics[width=1\textwidth]{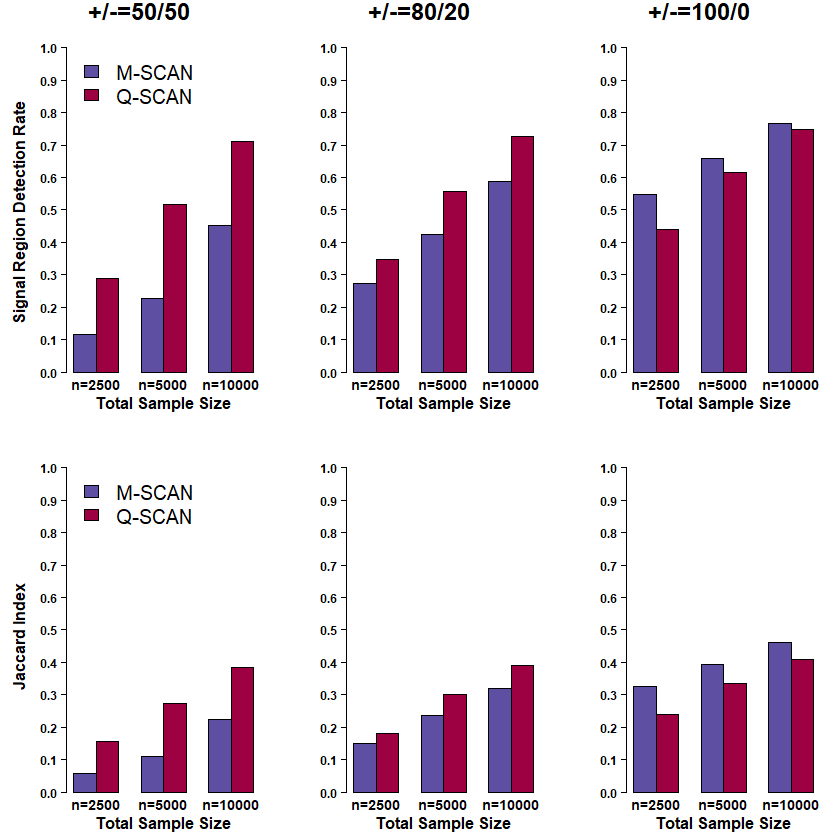}
\end{figure}

\pagebreak
\newpage
\thispagestyle{empty}

\begin{figure}[H]
\caption{Power and accuracy of estimated signal region comparisons using Q-SCAN and M-SCAN assuming the sparsity index $\xi=1/2$ . We evaluated power via the signal region detection rate and the Jaccard index defined in the simulation section. Both criteria were calculated at the family-wise error rate at 0.05 level. The number of variants in signal region $p_0$ was randomly selected between 50 and 80. The $s = p_0^\xi$ causal variants were selected randomly within each signal region. Each  causal variant has an effect size that is set as a decreasing function of MAF, $\beta=c\log_{10}$(MAF) and $c=0.30$.  From left to right, the plots consider settings in which the coefficients for the causal rare variants are 100\% positive (0\% negative), 80\% positive (20\% negative), and 50\% positive (50\% negative). We repeated the simulation for 1000 times. Q-SCAN and M-SCAN refer to the scan procedures using the scan statistics $\sum_{i \in I} U_i^2  - \E(\sum_{i \in I} U_i^2)/var(\sum_{i \in I} U_i^2)$ and $\big(\sum_{i \in I} U_i\big)^2/ var\big(\sum_{i \in I} U_i\big)$, respectively. In both two scan procedures, the range of variants number in searching windows was between 40 and 200.}
\centering
\label{fig:Q_SCAN_1_2}
\includegraphics[width=1\textwidth]{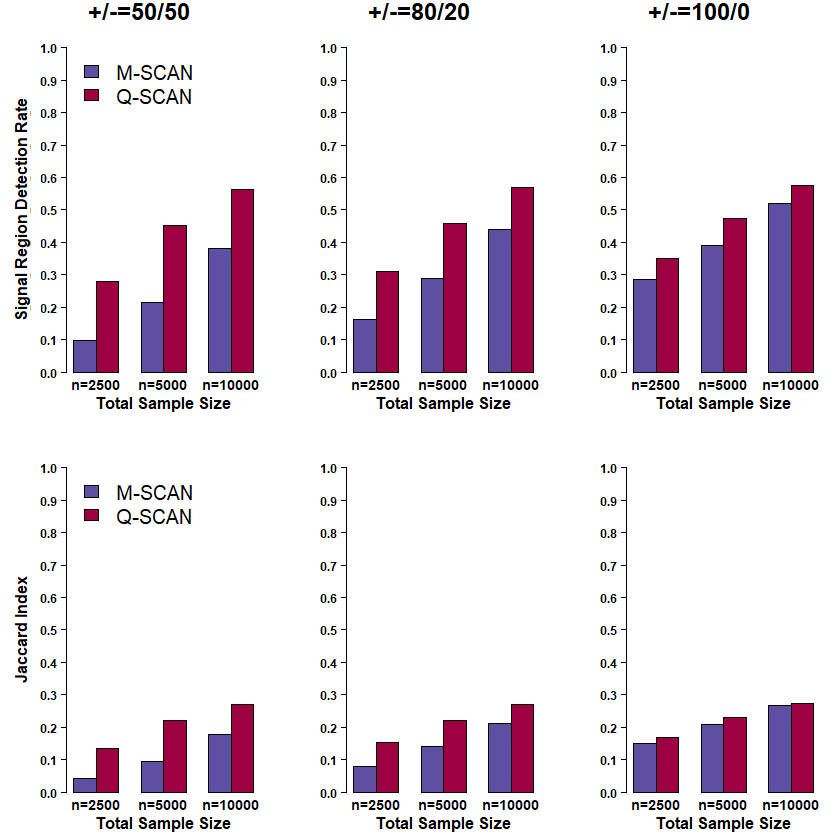}
\end{figure}

\pagebreak
\newpage
\clearpage
\thispagestyle{empty}

\pagebreak
\newpage
\clearpage
\thispagestyle{empty}

\begin{figure}[H]
\caption{Genetic landscape of the windows that are significantly associated with neutrophil counts on chromosome 1 among African Americans in the ARIC Whole Genome Sequencing Study. Three methods are compared: Q-SCAN, M-SCAN and 4 kb sliding window procedures using SKAT. A dot means that the sliding window at this location is significant using the method that the color of the dot represents. The physical positions of windows are based on build hg19.
}
\centering
\label{fig:Neu-landscape}
\includegraphics[width=0.85\textwidth]{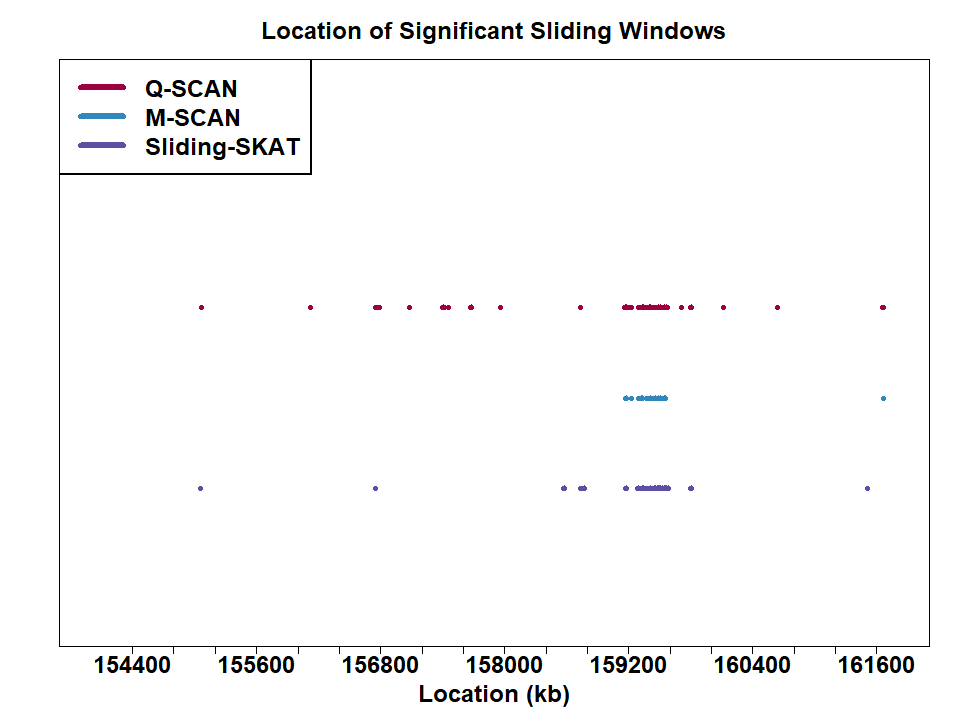}
\end{figure}

\end{document}